\documentclass[conference,a4paper]{IEEEtran}
\IEEEoverridecommandlockouts
\usepackage{stfloats,color}
\usepackage{amsfonts}
\usepackage[cmex10]{amsmath}
\usepackage{graphicx}
\usepackage{setspace}
\usepackage{cite}
\usepackage{array}
\usepackage{algorithmic}
\usepackage{mdwmath}
\usepackage{mdwtab}
\usepackage{url}
\usepackage[center]{caption}
\usepackage[font=footnotesize]{subfig}
\usepackage{fixltx2e}
\usepackage{url}
\usepackage{times}
\usepackage{epsfig}
\usepackage{latexsym}
\usepackage{epstopdf}
\usepackage{verbatim}
\usepackage{units}
\usepackage{amsthm}
\usepackage{mdwlist}

\theoremstyle{remark}
\newtheorem{theorem}{Theorem}
\newtheorem{lemma}[theorem]{Lemma}

\graphicspath{{./figures/}}

\begin{document}
\title{Cognitive Radio Networks with Probabilistic Relaying: Stable Throughput and Delay Tradeoffs 
\thanks{This paper was made possible by a NPRP grant 4-1034-2-385 from the
Qatar National Research Fund (a member of The Qatar Foundation). The
statements made herein are solely the responsibility of the authors.}
}

\author{\IEEEauthorblockN{Mahmoud Ashour$^\dag\!$, Amr A. El-Sherif$^\dag\!$, Tamer ElBatt$^*$, and Amr Mohamed$^\dag\!$}   

\IEEEauthorblockA{$^\dag$ Computer Science and Engineering Department, Qatar University, Doha, Qatar.\\ 
$^*\!\!$ Wireless Intelligent Networks Center (WINC), Nile University, Giza, Egypt.}
\{\tt m.ashour@qu.edu.qa, amr.elsherif, telbatt, amrm@ieee.org\}
}

\maketitle

\begin{abstract}
In this paper, we study and analyze fundamental throughput and delay tradeoffs in cooperative multiple access for cognitive radio systems. We focus on the class of randomized cooperative policies, whereby the secondary user (SU) serves either the queue of its own data or the queue of the primary user (PU) relayed data with certain service probabilities. Moreover, admission control is introduced at the relay queue, whereby a PU's packet is admitted to the relay queue with an admission probability. The proposed policy introduces a fundamental tradeoff between the delays of the PU and SU. Consequently, it opens room for trading the PU delay for enhanced SU delay and vice versa. Thus, the system could be tuned according to the demands of the intended application. Towards this objective, stability conditions for the queues involved in the system are derived. Furthermore, a moment generating function approach is employed to derive closed-form expressions for the average delay encountered by the packets of both users. The effect of varying the service and admission probabilities on the system's throughput and delay is thoroughly investigated. Results show that cooperation expands the stable throughput region. Moreover, numerical simulation results assert the extreme accuracy of the analytically derived delay expressions. In addition, we provide a criterion for the SU based on which it decides whether cooperation is beneficial to the PU or not. Furthermore, we show the impact of controlling the flow of data at the relay queue using the admission probability.

\end{abstract}


\IEEEpeerreviewmaketitle

\vspace{-2pt}

\section{Introduction}
\IEEEPARstart{T} \small{HE} extensive use of wireless communications recently collides with the shortage of resources required to establish communications. Spectrum scarcity coupled with the under-utilization of the licensed spectrum \cite{FCC} stimulated the introduction of the concept of cognitive radios \cite{Mitola,Haykin} aiming at exploiting the spectral holes. These holes are silence periods in which the spectrum is idle. The presence of such holes originates from the bursty nature of the sources, where the users who have legitimate access to the system, called primary users (PUs), do not always have data to transmit. That is why congitive radio networks have been gaining increasing worldwide interest. The main idea of cognitive radios resides in introducing cognitive secondary users (SUs) capable of sensing the spectrum and exploiting spectral holes for transmitting their packets. Thus, the spectral efficiency of the system is enhanced while simultaneously keeping the quality of service (QoS) requirements unviolated at the PUs \cite{akin2010effective}. 

Recently, cooperative communication in wireless networks has been widely investigated \cite{Tse,Kramer}. 
Cooperation has been made possible by the broadcast nature of wireless channels. As a result of such a nature, a single transmission can be received by different nodes within its range. Data lost over the direct link between a transmitter and its intended receiver is probably successfully received by a set of intermediate nodes. Each node in this set is considered a prospective relay that can deliver the lost data to their destinations. In \cite{Tse}, the authors outline several strategies employed by the cooperating radios including amplify-and-forward and decode-and-forward schemes. They develop performance characterizations in terms of outage events and associated outage probabilities.  
In \cite{Gamal}, cooperative transmission protocols for $N$ partners are proposed, where these protocols are evaluated using Zheng-Tse diversity-multiplexing tradeoff \cite{zheng2003diversity}. Sadek et al provided a symbol error rate analysis for decode and forward cooperation protocol in \cite{SER}. This analysis is used as a baseline for a relay selection mechanism developed and analyzed in \cite{ibrahim}. Cooperative communication can be also viewed as a way of implementing the notion of spatial diversity. Analogous to using multiple antennas to achieve spatial diversity in single communication links \cite{Foschini,Telatar}, the resources of multiple nodes can be exploited to induce a similar effect. Apparently, the previously listed works deal with cooperative communication from a physical layer prespective. However, we are interested in: $(i)$ employing cooperation at higher network layers and investigating its promises in terms of throughput and delay, $(ii)$ implementing cooperation in cognitive radio wireless networks.     

Incorporating cooperation into cognitive radio networks, the SUs not only seek idle time slots to transmit their own data, but they may also relay the PUs' lost packets. 
Thus, cooperation in cognitive radio networks can be viewed as a win-win situation. The SUs help the PUs deliver their packets to the destination. This helps in fulfilling the demand of the PUs and, hence, increasing the availability of slots in which SUs can transmit their own packets. For instance, in \cite{simeone2007stable}, power allocation at the SU, which has the capability of relaying the packets of the PU, is done with the objective of maximizing the stable throughput of the cognitive link for a fixed throughput selected by the primary link. In \cite{simeone2008spectrum}, the PU leases its own bandwidth for a fraction of time to a secondary network in exchange of appropriate gains attributed to cooperation. Multiple protocols are analyzed in \cite{Krikidis} which allow cooperation between a PU and a set of SUs. Perhaps an interesting point in \cite{Krikidis} is enabling simultaneous transmission of primary and secondary data using dirty-paper coding \cite{DPC}. In \cite{Sadek}, two protocols are developed and analyzed to implement cooperation in a system of $M$ source terminals, a single destination, and a single cognitive relay. Protocol-level cooperation is implemented in \cite{Ephremedis} among $N$ nodes in a wireless network, whereby each node is a source and a prospective relay at the same time. Performance gains in terms of stable throughput region and average delay are demonstrated.

In this paper, we consider a cognitive scenario in which the SU keeps two queues, one for its own packets and the other for the PU's relayed packets. Unlike the conventional relaying that assigns full priority to the relay queue, our prime objective is to develop a mathematical framework for the class of randomized cooperative policies that open room for accommodating cognitive radio systems supporting real-time, e.g., multimedia, and traffic with stringent QoS requirements, aka opportunistic real-time (ORT) \cite{Hossam}. Moreover, we take into account the QoS guarantees at the PU. 
Towards this objective, we propose and analyze a tunable randomized service cooperative policy with probabilistic relaying. According to the proposed policy, admission control is introduced at the relay queue, where a PU's packet that fails to reach the destination, is admitted to the relay queue with probability (w.p.) $p_{a}$ upon being successfully decoded by the SU. In addition, when the SU detects an idle time slot and decides to transmit, it serves either the queue of its own data w.p. $p_{q}$, or the relay queue w.p. $(1-p_{q})$. Consequently, we open room for trading the PU delay for enhanced SU delay and vice versa. Thus, the system could be tuned according to the demands of the intended applications running at both the PU and SU. Fundamental stable throughput and delay tradeoffs at both users are studied. The significance of the proposed policy lies in its tunability, whereby a variety of objectives could be realized via performing constrained optimizations over the degrees of freedom of the system represented by the admission probability to the relay queue, $p_{a}$, and the queue selection probability, $p_{q}$. Hereafter, we refer to the probabilistic queue selection by the term randomized service, while we refer to the admission control introduced at the relay queue by the term  probabilistic relaying. 

It is worth referring to the work done in \cite{pappas}, where a two-user cooperative scenario with admission control at the relay queue is considered. The authors are solely concerned with the derivation of the stable throughput region. Therefore, they do not differentiate between the two queues maintained by the cooperating terminal, i.e., the queue of own data and the queue of the relayed data. Unlike \cite{pappas}, we take into account the randomized service at the SU in the derivations of the stable throughput region. In addition, we provide a detailed analysis for the average delay encountered by the packets of both the PU and SU, which is out of the scope of \cite{pappas}.         
The main contributions of this work are summarized as follows:
\begin{enumerate}

\item We propose a randomized service cooperative policy with probabilistic relaying that enables trading the PU delay for enhanced SU delay and vice versa, depending on the application and system QoS constraints.

\item Under the proposed policy, the stable throughput region of the system is derived. 
Moreover, we derive closed-form expressions for the average delay experienced by the packets of both users. Furthermore, the effect of varying $p_{q}$ and $p_{a}$ on the system's throughput and delay is thoroughly investigated.

\item Extensive simulations are conducted to validate and, show the accuracy of, the obtained analytical results.
  
\item A fundamental tradeoff between the average delay and throughput of both users is studied and analyzed. At any given point within the stable throughput region of the system, we solve for the optimal values of $(p_{q},p_{a})$ that minimize the average delay for the PU and SU. Moreover, we study the tradeoff between the delays of the PU and SU, with emphasis on the role of $p_{a}$ in this tradeoff at different values of $p_{q}$.

\item We provide a criterion for the SU based on which it decides whether cooperation is beneficial to the PU or not. Also, we clearly define the gains of cooperation. In addition, we show the potential of using the admission probability as a flow regulator at the relay queue.
\end{enumerate}    

The rest of this paper is organized as follows. Section \ref{system model} presents the system model along with the implemented cooperation strategy. Section \ref{stable throughput region} presents the derivation of the stable throughput region of the system. The average delay characterization of the system is provided in section \ref{delay}. Numerical results are then presented in section \ref{numerical results}. Finally, a concluding discussion that summarizes the key insights and design guidelines inspired by our theoretical findings is presented in \ref{conclusion}.

\begin{figure}[t]
\begin{center}
\includegraphics[width=1\columnwidth , height=0.32\columnwidth]{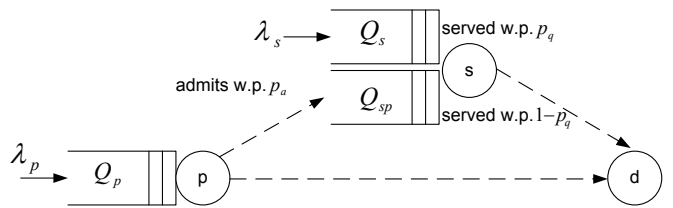}
\caption{Cognitive radio network under consideration.} \label{Fig1}
\end{center}
\vspace{-7mm}
\end{figure}

\section{System Model}\label{system model}
We consider the cognitive radio system shown in Fig. \ref{Fig1}. The system comprises a PU and a SU equipped with infinite capacity buffers, transmitting their packets to a common destination $d$. Time is slotted, and the transmission of a packet takes exactly one time slot. Source burstiness is taken into account through modelling the arrivals at the PU and SU as Bernoulli processes with average rates denoted by $\lambda_{p}$ and $\lambda_{s}$ (packets per slot), respectively, whereby the typical values of $\lambda_{p}$ and $\lambda_{s}$ lie in the interval $[0,1]$. The arrival processes at both users are independent of each other, and are independent and identically distributed (i.i.d) across time slots. 

The channel quality is determined by the average channel reception probability. It is the probability that a packet is decoded without error. A successful transmission requires receiving the entire packet without error, otherwise, the packet is discarded. The channel gain and noise processes are assumed stationary. Thus, the channel reception probability is a constant value between $0$ and $1$. Moreover, we assume that the SU performs perfect sensing. Thus, the system is contention-free, since at most one user is allowed to transmit in a given slot. Hence, the sole cause for packet loss is the channel impairments since no collisions are allowed. These channel impairments are typically caused by fading, shadowing, signal attenuation and additive noise. The event that causes packet loss is the channel outage event, which is characterized as the signal-to-noise ratio (SNR) at the receiving node being below a pre-specified threshold. This threshold is the minimum value of the SNR required by the receiver to perform an error-free decoding. Let $f_{pd}$, $f_{sd}$, and $f_{ps}$ denote the probability of success between the PU and destination, the SU and destination, and the PU and SU, respectively. It is assumed throughout the paper that $f_{pd}<f_{sd}$. Acknowledgements (ACKs) sent by the destination, and the SU for overheard primary packets, are assumed instantaneous and can be heard by all nodes error-free.

Next, we present the queueing model of the system followed by the description of the cooperation strategy.

\subsection{Queueing Model}
There are three queues involved in the system analysis, as shown in Fig. \ref{Fig1}. They are described as follows:
\begin{itemize}
\item $Q_{p}$: stores the packets of the PU corresponding to the exogenous Bernoulli arrival process with rate $\lambda_{p}$.

\item $Q_{sp}$: stores the packets at the SU, overheard from the PU.

\item $Q_{s}$: stores the packets of the SU corresponding to the exogenous Bernoulli arrival process with rate $\lambda_{s}$.  
\end{itemize}

The instantaneous evolution of the length of queue $k$ is captured as
\begin{align}\label{queue evolution}
Q_{k}^{t+1}= [Q_{k}^{t} - Y_{k}^{t}]^{+} + X_{k}^{t}, ~~k \in \{ p,sp,s \}
\end{align}  
where $[x]^{+}=\text{max}(x,0)$, and $Q_{k}^{t}$ denotes the number of packets in the $k$th queue at the beginning of the $t$th time slot. The binary random variables taking values either $0$ or $1$, $Y_{k}^{t}$ and $X_{k}^{t}$, denote the departures and arrivals corresponding to the $k$th queue in the $t$th time slot, respectively. 

\subsection{Cooperation Strategy}\label{cooperation strategy}
The proposed cooperative scheme is described as follows:
\begin{enumerate}
\item The PU transmits a packet whenever $Q_{p}$ is non empty.

\item If the packet is successfully decoded by the destination, it broadcasts an ACK that can be heard by both users in the network. Thus, the packet exits the system.

\item If the packet is not successfully received by the destination, yet, successfully decoded by the SU, $Q_{sp}$ either buffers the packet w.p. $p_{a}$ or discards it w.p. $(1-p_{a})$. This constitutes the probabilistic relaying admission policy. 

\item If the packet is buffered in $Q_{sp}$, the SU sends back an ACK to announce successful reception of the PU's packet. Therefore, the packet is dropped from $Q_{p}$ and becomes the responsibility of the SU to deliver to the destination.

\item If the packet is neither successfully received by the destination nor decoded by the SU and admitted to $Q_{sp}$, it is kept at $Q_{p}$ for retransmission in the next time slot.
 
\item When the PU is idle, the SU transmits a packet from either $Q_{s}$ or $Q_{sp}$ w.p. $p_{q}$ and ($1-p_{q}$), respectively.

\item If the packet is successfully decoded by the destination, it sends back an ACK and the packet exits the system. Otherwise, it is kept at its queue for later retransmission. 

\end{enumerate}

It is worth noting from the description of the proposed policy that the system at hand is non work-conserving. A system is considered work-conserving if it does not idle whenever it has packets \cite{Wolff}. However, in our system, one case violates this condition, which arises when the SU detects a slot in which the PU is idle, and it randomly selects to transmit a packet from one of its queues which turns out to be empty, while the other queue is non-empty. Accordingly, the slot would go idle and be wasted despite the system having packets awaiting transmission. Clearly, this results in a degradation in the system performance. Nevertheless, we can extend it to a more flexible work-conserving version of the proposed policy that exploits the resources efficiently without the risk of wasting slots. However, its delay analysis is notoriously complex since it involves deriving the moment generating function of the joint lengths of the three queues in the system. Thus, we resort to the non work-conserving policy for its mathematical tractability. Consequently, we derive closed-form expressions for the expected packet delay, formulate and solve, analytically, optimization problems with the objective of minimizing delay at both users.   

\section{Stable Throughput Region}\label{stable throughput region}
The system is considered stable when all of its queues are stable. Queue stability is loosely defined as having a bounded queue size, i.e., the number of packets in the queue does not grow to infinity \cite{Sadek}. In this section, we characterize the stable throughput region of the system. Moreover, we distill valuable insights related to the effect of tuning the system parameters, $(p_{q},p_{a})$, on the stability region of the system.

\begin{theorem}\label{Thm1}
The stable throughput region for the system in Fig. \ref{Fig1} under the proposed randomized service policy with probabilistic relaying, for a fixed value of $(p_{q},p_{a})$, is given by
\begin{align}\label{stable throughput region for a given value of a}
\mathbf{R}=\bigg\{& (\lambda_{p},\lambda_{s}): \lambda_{s}<p_{q}f_{sd} \left[ 1 - \frac{\lambda_{p}}{f_{pd}+p_{a}f_{ps}(1-f_{pd})} \right], \notag \\ 
&\text{for}~ \lambda_{p}< \frac{f_{sd}(1-p_{q})[f_{pd} + p_{a}f_{ps}(1-f_{pd})]}
{f_{sd}(1-p_{q})+p_{a}f_{ps}(1-f_{pd})}  \bigg\}
\end{align}
\end{theorem}
\begin{proof}
We use Loynes' theorem \cite{Loynes} to establish the stability of each queue. The theorem states that if the arrival and the service processes of a queue are stationary, then the queue is stable if and only if the arrival rate is strictly less than the service rate. 

\begin{itemize}
\item For $Q_{p}$ stability, the following condition must be satisfied
\begin{align}\label{Q1}
\lambda_{p} < \mu_{p}
\end{align}
where $\mu_{p}$ denotes the service rate of $Q_{p}$.
A packet departs $Q_{p}$ if it is successfully received by the destination or is decoded by the SU and is admitted to its relay queue.
Thus, $\mu_{p}$ is given by
\begin{align}\label{mu_p}
\mu_{p}\! = \! 1\! - \! (1\! - \! f_{pd})(1\! - \! p_{a}f_{ps})\! =\! f_{pd} \! + \! p_{a}f_{ps}(1-f_{pd})
\end{align} 
\item For $Q_{sp}$ stability, the following condition must be satisfied
\begin{align}\label{Q21}
p_{a}f_{ps}(1-f_{pd})\frac{\lambda_{p}}{\mu_{p}}<\left[ 1-\frac{\lambda_{p}}{\mu_{p}} \right](1-p_q)f_{sd}
\end{align}
A PU's packet is buffered at $Q_{sp}$ if an outage occurs in the link between the PU and the destination which happens w.p. ($1-f_{pd}$), yet, no outage occurs in the link between the PU and the SU which happens w.p. $f_{ps}$, and the packet is admitted to $Q_{sp}$ which occurs w.p. $p_{a}$, while $Q_{p}$ is not empty which has a probability of $ \lambda_{p} / \mu_{p} $. This explains the left hand side of (\ref{Q21}) which is the rate of packet arrivals to the SU relay queue. The right hand side represents the service rate seen by the packets of $Q_{sp}$. A packet departs the relay queue if $Q_{p}$ is empty, $Q_{sp}$ is selected to transmit a packet, and there is no outage in the link between the SU and the destination. 
Rearranging the terms of the above equation yields the following condition on the maximum achievable arrival rate at the PU
\begin{align}\label{Q21_modified}
\lambda_{p}< \left[ \frac{f_{sd}(1-p_{q})}{f_{sd}(1-p_{q})+p_{a}f_{ps}(1-f_{pd})} \right]\mu_{p} 
\end{align} 
Comparing (\ref{Q1}) and (\ref{Q21_modified}), it becomes clear that (\ref{Q21_modified}) provides a tighter bound on $\lambda_{p}$ due to the multiplication of $\mu_{p}$ by a term which is less than one.

\item For $Q_{s}$ stability, the following condition must be satisfied
\begin{align}\label{Q22}
\lambda_{s}<p_{q}f_{sd} \left[ 1 - \frac{\lambda_{p}}{\mu_{p}} \right]
\end{align}
Using the same rationale, a packet departs $Q_{s}$ if $Q_{p}$ is empty, $Q_{s}$ is selected to transmit a packet, and there is no outage in the link between the SU and the destination. This explains the service rate seen by the packets of $Q_{s}$ given in the right hand side of (\ref{Q22}).
\end{itemize}   

The stability conditions given by (\ref{Q21_modified}) and (\ref{Q22}) establish the result in (\ref{stable throughput region for a given value of a}). 
\end{proof}

Next, we study and analyze the sensitivity of the stable throughput region of the system to changes in both $p_{q}$ and $p_{a}$. We begin first by investigating the effect of varying $p_{q}$ while keeping $p_{a}$ constant, followed by the other way round, i.e., varying $p_{a}$ while keeping $p_{q}$ fixed. 

\begin{lemma}\label{Lemma1}
The maximum achievable arrival rate at the PU, $\lambda_{p}$, decreases monotonically with $p_{q}$. Conversely, for a fixed $\lambda_{p}$, the maximum achievable arrival rate at the SU, $\lambda_{s}$, increases monotonically with $p_{q}$.
\end{lemma}

\begin{proof}
From the system stability conditions, the maximum achievable $\lambda_{p}$, that defines the boundary of the stable throughput region for a given $(p_{q},p_{a})$, is given by (\ref{Q21_modified}). Taking the derivative of (\ref{Q21_modified}) with respect to (w.r.t.) $p_{q}$ yields 
\begin{equation}\label{lambda_p vs p1}
\frac{\partial \lambda_{p}}{\partial p_{q}}=\frac{-p_{a}f_{sd}f_{ps}(1-f_{pd})\mu_{p}}
{[f_{sd}(1-p_{q})+p_{a}f_{ps}(1-f_{pd})]^{2}}
\end{equation}  
Since $p_a$, $f_{sd}$, $f_{ps}$, $f_{pd}$, and $\mu_{p}$ are all positive numbers less than one, we conclude from (\ref{lambda_p vs p1}) that $\frac{\partial \lambda_{p}}{\partial p_{q}}$ is negative definite irrespective of the choice of $p_{a}>0$. This establishes that the maximum achievable $\lambda_{p}$ monotonically decreases with $p_{q}$.

On the other hand, for a fixed $\lambda_{p}$, the maximum achievable $\lambda_{s}$, that defines the boundary of the stable throughput region for a given value of $(p_{q},p_{a})$, is given by (\ref{Q22}). Taking the derivative of (\ref{Q22}) w.r.t. $p_{q}$ yields
\begin{equation}\label{lambda_s vs p1}
\frac{\partial \lambda_{s}}{\partial p_{q}}=f_{sd} \left[ 1- \frac{\lambda_{p}}{\mu_{p}} \right]
\end{equation}
The stability condition provided in (\ref{Q1}) guarantees that the utilization factor of $Q_{p}$, $\frac{\lambda_{p}}{\mu{p}}$, is less than one. Thus, it can be obviously seen from (\ref{lambda_s vs p1}) that $\frac{\partial \lambda_{s}}{\partial p_{q}}$ is positive definite irrespective of the choice of $p_{a}$. This establishes that, for a fixed $\lambda_{p}$, the maximum achievable $\lambda_{s}$ monotonically increases with $p_{q}$.
\end{proof}

\begin{lemma}\label{Lemma2}
The maximum achievable arrival rate at the PU, $\lambda_{p}$, increases monotonically with $p_{a}$ if $p_{q}$ lies in the interval $\left( 0,1-\frac{f_{pd}}{f_{sd}}\right)$, and decreases monotonically with $p_{a}$ if $p_{q}$ lies in the interval $\left(1-\frac{f_{pd}}{f_{sd}},1\right)$. However, for a fixed $\lambda_{p}$, the maximum achievable arrival rate at the SU, $\lambda_{s}$, increases monotonically with $p_{a}$, irrespective of the choice of $p_{q}$.   
\end{lemma}

\begin{proof}
Towards proving this result, we follow the same footsteps of the proof of Lemma \ref{Lemma1}. 
Taking the derivative of (\ref{Q21_modified}) w.r.t. $p_{a}$ yields 
\begin{equation}\label{lambda_p vs p2}
\frac{\partial \lambda_{p}}{\partial p_{a}}=(1-p_{q})(1-f_{pd})f_{ps}f_{sd}
\frac{  f_{sd}(1-p_{q})-f_{pd} }
{\left[ (1-p_{q})f_{sd}+p_{a}f_{ps}(1-f_{pd}) \right]^{2}}
\end{equation}  
Since $p_q$, $f_{pd}$, $f_{ps}$, and $f_{sd}$ are all positive numbers less than one, we conclude from (\ref{lambda_p vs p2}) that the behaviour of the maximum achievable $\lambda_{p}$ is governed by the term $f_{sd}(1-p_{q})-f_{pd}$. Solving for the value of $p_{q}$ that renders the maximum achievable $\lambda_{p}$ insensitive to variations in $p_{a}$, i.e., $\frac{\partial \lambda_{p}}{\partial p_{a}}=0$, we get
\begin{equation}\label{cut off}
p_{q}=1-\frac{f_{pd}}{f_{sd}}
\end{equation}
Evidently, it can be seen that if $p_{q}<1-\frac{f_{pd}}{f_{sd}}$, then $\frac{\partial \lambda_{p}}{\partial p_{a}}$ becomes positive, which implies that the maximum achievable $\lambda_{p}$ increases monotonically with $p_{a}$. On the other hand, $\frac{\partial \lambda_{p}}{\partial p_{a}}$ is negative, implying that the maximum achievable $\lambda_{p}$ decreases monotonically with $p_{a}$, if $p_{q}>1-\frac{f_{pd}}{f_{sd}}$.

At the SU side, taking the derivative of (\ref{Q22}) w.r.t. $p_{a}$ yields
\begin{equation}\label{lambda_s vs p2}
\frac{\partial \lambda_{s}}{\partial p_{a}}=
\frac{p_{q}p_{a}f_{sd}f_{ps}(1-f_{pd})\lambda_{p}}{\left[ f_{pd}+p_{a}f_{ps}(1-f_{pd}) \right]^{2}}
\end{equation} 
Since $p_q$, $p_{a}$, $f_{pd}$, $f_{ps}$, $f_{sd}$, and $\lambda_{p}$ are all positive numbers less than one, we conclude from (\ref{lambda_s vs p2}) that $\frac{\partial \lambda_{s}}{\partial p_{a}}$ is always positive definite, irrespective of the choice of $(p_{q},p_{a})$. Thus, it has been established that, for a fixed $\lambda_{p}$, the maximum achievable $\lambda_{s}$ increases monotonically with $p_{a}$.  
\end{proof}

Figs. \ref{Fig2} and \ref{Fig3} illustrate the results obtained in Lemmas \ref{Lemma1} and \ref{Lemma2}. 
In an attempt to check how the stable throughput region behaves in response to variations in $p_{q}$,
we plot in Fig. \ref{Fig2} the stable throughput region of the system under the proposed policy fixing $p_{a}=1$ and varying $p_{q}$. 
Hereafter, the system parameters are chosen as follows: $f_{pd}=0.3$, $f_{ps}=0.4$, and $f_{sd}=0.8$, unless otherwise stated. According to this figure, we depict the effect of the probability $p_{q}$ on the stability region of the system. It can be realized that increasing the value of $p_{q}$ decreases the maximum achievable arrival rate at the PU, $\lambda_{p}$. On the contrary, increasing $p_{q}$ results in an increase in the maximum achievable arrival rate at the SU, $\lambda_{s}$, for every feasible $\lambda_{p}$. This result is intuitive, since increasing the value of $p_{q}$ gives more chance for transmitting the SU own packets as opposed to the PU's relayed packets. This, in turn, reduces the degree of cooperation the PU experiences from the SU and, hence, the maximum achievable $\lambda_{p}$ decreases. On the other hand, since the SU own packets are more likely to be transmitted, the system can sustain higher values of $\lambda_{s}$. Thus, we conclude that increasing $p_{q}$ is always in favor of the SU as opposed to the PU. 

\begin{figure}[t]
\begin{center}
\includegraphics[width=1\columnwidth , height=0.68\columnwidth]{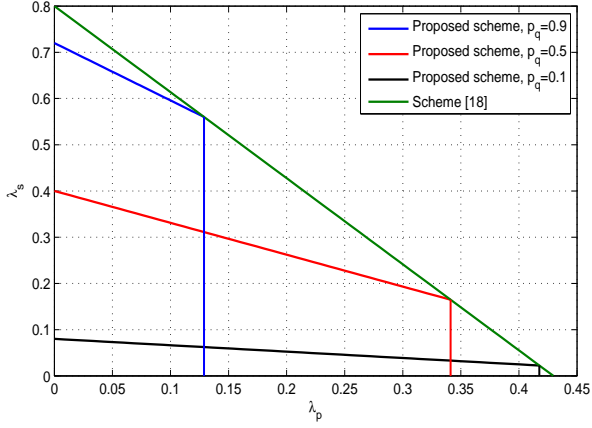}
\caption{Stable throughput region at $p_{a}=1$ for different values of $p_{q}$.} \label{Fig2}
\end{center}
\vspace{-7mm}
\end{figure}
In Fig. \ref{Fig3}, we examine the behavior of the stable throughput region of the system in response to variations in $p_{a}$ at the ``phase transition'' value of $p_{q}$ provided in (\ref{cut off}) which equals $0.625$ for the previously given values for $f_{pd}$, $f_{ps}$ and $f_{sd}$. It can be noticed from the figure that the maximum achievable $\lambda_{p}$ is insensitive to variations in $p_{a}$, i.e., it is constant for all values of $p_{a}$. However, at a fixed $\lambda_{p}$, the maximum achievable $\lambda_{s}$ increases with the increase of $p_{a}$. 

The result obtained in Lemma \ref{Lemma2} is further clarified via Figs. \ref{Fig4} and \ref{Fig5}, where we plot the maximum achievable $\lambda_{p}$ given by (\ref{Q21_modified}) and the maximum achievable $\lambda_{s}$ given by (\ref{Q22}) versus $p_{a}$, respectively, at different values of $p_{q}$. It is shown in Fig. \ref{Fig4} that the maximum achievable $\lambda_{p}$ increases monotonically with $p_{a}$ as long as $p_{q}<1-\frac{f_{pd}}{f_{sd}}$. Conversely, it decreases monotonically with $p_{a}$ for $p_{q}>1-\frac{f_{pd}}{f_{sd}}$ while remaining constant at $p_{q}=1-\frac{f_{pd}}{f_{sd}}$. Perhaps an intuitive explanation for this behavior is the following: If the channel quality between the PU and the destination is much worse than that between the SU and the destination, i.e., $f_{pd}<<f_{sd}$, then over almost the entire range of $p_{q} \in (0,1)$, the PU's throughput is enhanced via cooperation, i.e., having more packets getting relayed by the SU enhances the PU's throughput. However, if the channel between the PU and the destination is at least as good as the channel between the SU and the destination, i.e., $f_{pd}\simeq f_{sd}$, then it is always in the interest of the PU to retransmit its lost packets rather than getting them relayed via the SU, i.e., rejecting more packets at $Q_{sp}$ enhances the PU's throughput. Another interesting way of explaining this result comes through rearranging (\ref{cut off}). The PU benefits from cooperation as long as $(1-p_{q})f_{sd}>f_{pd}$, that is, the success probability over the relay-destination link is greater than that of the PU-destination link.  

Back to Fig. \ref{Fig4}, it can be noticed that at a fixed $p_{a}$, the system can sustain higher values of $\lambda_{p}$ at lower values of $p_{q}$, which is the result obtained in Lemma \ref{Lemma1} and shown in Fig. \ref{Fig2}. In addition, one can notice that the degradation in the PU's throughput with the increase of $p_{q}$ decreases at lower values of $p_{a}$. These results stimulate thinking of $p_{a}$ as an effective parameter that could be tuned to tailor the performance of the system to the demands of the intended application. 

\begin{figure}[t]
\begin{center}
\includegraphics[width=1\columnwidth , height=0.68\columnwidth]{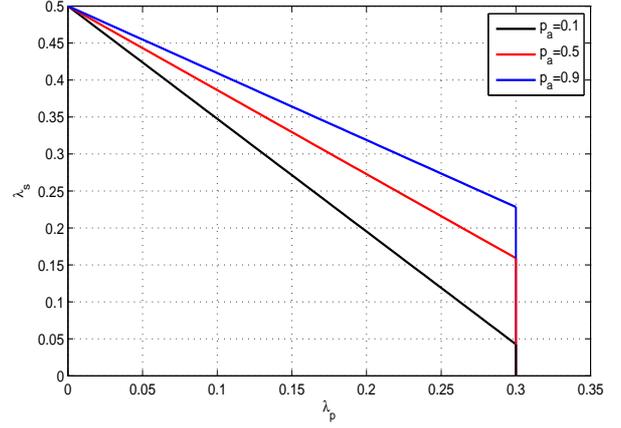}
\caption{Stable throughput region at $p_{q}=1-\frac{f_{pd}}{f_{sd}}$ for different values of $p_{a}$.} \label{Fig3}
\end{center}
\vspace{-7mm}
\end{figure}
On the other hand, we plot in Fig. \ref{Fig5} the maximum achievable $\lambda_{s}$ versus $p_{a}$ at a fixed $\lambda_{p}$ chosen to be $0.2$. It can be depicted that at a fixed $\lambda_{p}$, the maximum achievable $\lambda_{s}$ increases monotonically with $p_{a}$ independent of the choice of $p_{q}$. Therefore, it is clear that the SU is always benefiting from increasing $p_{a}$. This is attributed to the increase in the availability of time slots in which the PU's queue is empty, since at higher values of $p_{a}$, more packets are enqueued in $Q_{sp}$ and, hence, dropped from $Q_{p}$. Thus, the SU's packets are more likely to be transmitted. Moreover, at a fixed $p_{a}$, it can be seen from Fig. \ref{Fig5} that the system sustains higher values of $\lambda_{s}$ at higher values of $p_{q}$, which again emphasizes the result obtained in Lemma \ref{Lemma1} and shown in Fig. \ref{Fig2}.        

\begin{figure}[t]
\begin{center}
\includegraphics[width=1\columnwidth , height=0.68\columnwidth]{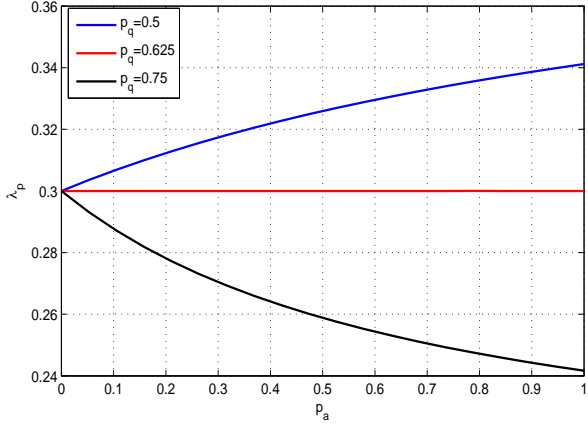}
\caption{Maximum achievable $\lambda_{p}$ versus $p_{a}$ for different values of $p_{q}$.} \label{Fig4}
\end{center}
\vspace{-7mm}
\end{figure}
 
After thoroughly investigating the effect of $p_{q}$ and $p_{a}$ on the stability region of the system, we present next a complete characterization of the stable throughput region of the system under the proposed policy by taking the union of (\ref{stable throughput region for a given value of a}) over all possible values of $(p_{q},p_{a})$.   
\begin{theorem}
The union of the stability regions given by (\ref{stable throughput region for a given value of a}) over all possible values of $(p_{q},p_{a})$ is the same as that of any work-conserving cooperative scheme, e.g., the one derived in \cite{Ephremedis}, and is given by 
\begin{equation}
\lambda_{s} < f_{sd}- \left[ \frac{f_{sd}+f_{ps}(1-f_{pd})}{f_{pd}+f_{ps}(1-f_{pd})} \right] \lambda_{p}
\end{equation}
\end{theorem}

\begin{proof}
The stable throughput region of the system for a fixed value of the pair $(p_{q},p_{a})$ is derived in Theorem \ref{Thm1} and is given by (\ref{stable throughput region for a given value of a}). To determine the union of the stability regions, we need to take the union over all possible values of $(p_{q},p_{a})$. A method used to characterize this union has been proposed in \cite{Sadek} in an analogous problem. It resorts to solving a constrained optimization problem to find the maximum feasible $\lambda_{s}$ corresponding to each feasible $\lambda_{p}$. Proceeding with this same objective, we make use of the result obtained in Lemma \ref{Lemma2}, where it has been established that the maximum achievable $\lambda_{s}$, at a fixed $\lambda_{p}$, increases monotonically with $p_{a}$ irrespective of the choice of $p_{q}$. This suggests that for obtaining the maximum over all attainable $\lambda_{s}$ at a fixed $\lambda_{p}$, we fix $p_{a}=1$ and optimize over $p_{q}$. Consequently, we employ the result presented in \cite{Ashour} (Section III-Theorem $4$), where our problem boils down at $p_{a}=1$ to the case presented therein.     
\end{proof}

It is worth noting that the overall stable throughput region of the system is shown in Fig. \ref{Fig2}.
Proceeding with the system analysis, it remains to study and analyze an important performance metric which is the expected packet delay.

\begin{figure}[t]
\begin{center}
\includegraphics[width=1\columnwidth , height=0.68\columnwidth]{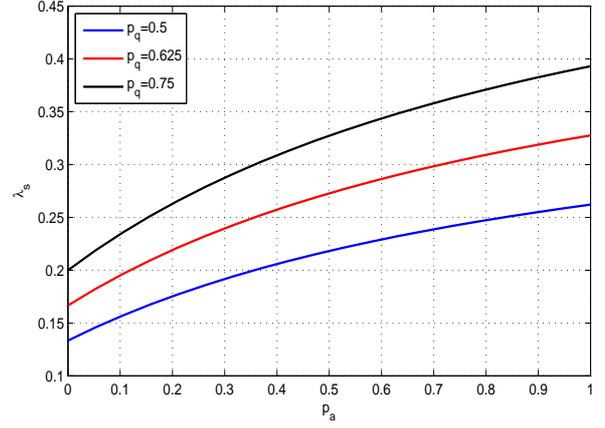}
\caption{Maximum achievable $\lambda_{s}$ versus $p_{a}$ at $\lambda_{p}=0.2$ for different values of $p_{q}$.} \label{Fig5}
\end{center}
\vspace{-7mm}
\end{figure}

\section{Average Delay Characterization}\label{delay} 
In this section, we perform the delay analysis of the system under the proposed scheme. Closed-form expressions are derived for the average delay encountered by the packets of the PU as well as the SU. Furthermore, the effect of tuning $p_{q}$ and $p_{a}$ on the average delay seen by the packets of both users is investigated. 

\begin{theorem}\label{Thm2}
The average delay encountered by the packets of the PU and SU, $D_{p}$ and $D_{s}$, respectively, under the proposed scheme, are given by 
\begin{equation}\label{Dp}
D_{p}= \frac{N_{p} + N_{sp}}{\lambda_{p}}
\end{equation}
\begin{equation}\label{Ds}
D_{s}= \frac{N_{s}}{\lambda_{s}} \quad 
\end{equation}
where $N_{p}$ and $N_{sp}$, the average lengths of $Q_{p}$ and $Q_{sp}$, respectively, are given by
\begin{equation}\label{pk}
N_{p}=\frac{-\lambda_{p}^{2} + \lambda_{p}}{f_{pd}+p_{a}f_{ps}(1-f_{pd})-\lambda_{p}}
\end{equation} 

\begin{equation}\label{Nsp}
N_{sp}= \frac{m \lambda_{p}^{2} + n \lambda_{p}}{\alpha \lambda_{p}^{2} + \beta \lambda_{p} + \gamma}
\end{equation} 
where
\begin{align}
\begin{split}
& m = p_{a}f_{ps}(1-f_{pd}) \bigg[ \frac{(1-p_{q})f_{sd}-f_{pd}}{f_{pd}+p_{a}f_{ps}(1-f_{pd})} \\
&\phantom{m = p_{a}f_{ps}(1-f_{pd})\bigg[}- (1-p_{q})f_{sd} - p_{a}f_{ps}(1-f_{pd}) \bigg ] 
\end{split} \notag \\
\begin{split}
n = p_{a}f_{ps}(1-f_{pd})\left[ f_{pd}+p_{a}f_{ps}(1-f_{pd}) \right ]
\end{split}\notag \\
\begin{split}
\alpha = (1-p_{q})f_{sd} + p_{a}f_{ps}(1-f_{pd})
\end{split}\notag \\
\begin{split}
\beta = \left[ f_{pd}+p_{a}f_{ps}(1-f_{pd}) \right] \left[ -2(1-p_{q})f_{sd}-p_{a}f_{ps}(1-f_{pd}) \right]
\end{split}\notag \\
\begin{split}
\gamma = (1-p_{q})f_{sd} \left[ f_{pd}+p_{a}f_{ps}(1-f_{pd}) \right]^{2}
\end{split}
\end{align}
and $N_{s}$, the average length of $Q_{s}$, is given by
\begin{equation}\label{Ns}
N_{s}=\frac{\lambda_{p}\lambda_{s}A + (\lambda_{s}^{2}-\lambda_{s})B(B+\lambda_{p})}{BC}
\end{equation}
where
\begin{align}
\begin{split}
A=p_{q}f_{sd}[f_{pd}+p_{a}f_{ps}(1-f_{pd})-1]
\end{split} \notag \\ 
\begin{split}
B=f_{pd}+p_{a}f_{ps}(1-f_{pd})-\lambda_{p}
\end{split} \notag \\
\begin{split}
C=(\lambda_{s}-p_{q}f_{sd})[f_{pd}+p_{a}f_{ps}(1-f_{pd})] + p_{q}f_{sd}\lambda_{p}
\end{split}
\end{align}
\end{theorem}

\begin{proof}
We start by computing the average delay of the SU's packets followed by the calculation of the average delay of the PU's packets. 

By applying Little's law on $Q_{s}$, we obtain $D_{s}$ exactly as given by (\ref{Ds}). Thus, it remains to calculate $N_{s}$, the average length of $Q_{s}$. The dependence of the service processes at both $Q_{s}$ and $Q_{sp}$ on the state of $Q_{p}$ is inherent from the concept of cognitive radios.  
It is worth noting that the non work-conserving behaviour of the proposed strategy makes the delay analysis of the system mathematically tractable, since $Q_{s}$ and $Q_{sp}$ become independent, i.e., having independent arrivals and departures. To analyze the average delays at different queues, we resort to the moment generating function approach \cite{33}. The moment generating function of the joint lengths of $Q_{p}$ and $Q_{s}$ is defined as 
\begin{align}\label{MGF}
G(x,y)&= \lim_{t \rightarrow \infty} \mathbf{E} \left[ x^{Q_{p}^{t}} y^{Q_{s}^{t}} \right] \notag \\
&= \lim_{t \rightarrow \infty} \displaystyle \sum_{i=0}^{\infty} \sum_{j=0}^{\infty} x^{i} y^{j} \mathbf{P} \left[ Q_{p}^{t}=i, Q_{s}^{t}=j \right]
\end{align} 
where $\mathbf{E}$ and $\mathbf{P}$ denote the statistical expectation and the probability operators, respectively. To illustrate the motivation of employing the moment generating function approach in our delay analysis, we take the derivative of (\ref{MGF}) w.r.t. $y$ which yields $G_{y}(x,y)$ that is given by
\begin{equation}
G_{y}(x,y)=\lim_{t \rightarrow \infty} \displaystyle \sum_{i=0}^{\infty} \sum_{j=1}^{\infty} j x^{i} y^{j-1} \mathbf{P} \left[ Q_{p}^{t}=i, Q_{s}^{t}=j \right]
\end{equation}
Substituting by $x=y=1$ in the above equation, it becomes clear that
\begin{align}
G_{y}(1,1)&=\lim_{t \rightarrow \infty} \displaystyle \sum_{j=1}^{\infty} j \displaystyle \sum_{i=0}^{\infty} \mathbf{P}
\left[ Q_{p}^{t}=i, Q_{s}^{t}=j \right] \notag \\
&=\lim_{t \rightarrow \infty} \displaystyle \sum_{j=1}^{\infty} j \mathbf{P}
\left[Q_{s}^{t}=j \right] = N_{s}
\end{align}
Thus, the sequence of characterizing $N_{s}$ goes as follows. First, we derive $G(x,y)$, then take its derivative w.r.t. $y$ and put $x=y=1$.

Proceeding with the derivation of $G(x,y)$, we make use of the queue evolution form provided by (\ref{queue evolution}). Thus, we have
\begin{align}
&\mathbf{E} \left[ x^{Q_{p}^{t+1}} y^{Q_{s}^{t+1}} \right] =
\mathbf{E} \left[ x^{(Q_{p}^{t}-Y_{p}^{t}+X_{p}^{t})} y^{(Q_{s}^{t}-Y_{s}^{t}+X_{s}^{t})} \right] 
\notag \\
&= (\lambda_{p}x+1-\lambda_{p})(\lambda_{s}y+1-\lambda_{s}) 
\mathbf{E} \left[ x^{(Q_{p}^{t}-Y_{p}^{t})} y^{(Q_{s}^{t}-Y_{s}^{t})} \right]
\end{align}
This follows from the independent arrival processes at $Q_{p}$ and $Q_{s}$, that yield independent Bernoulli distributed random variables, $X_{p}^{t}$ and $X_{s}^{t}$, which produce moment generating functions of $(\lambda_{p}x+1-\lambda_{p})$ and $(\lambda_{s}y+1-\lambda_{s})$, respectively. Expanding the above equation, we have
\begin{align}\label{expansion}
&\mathbf{E} \left[ x^{Q_{p}^{t+1}} y^{Q_{s}^{t+1}} \right]= \qquad \qquad \qquad \qquad \qquad \qquad \qquad \notag \\ &(\lambda_{p}x+1-\lambda_{p})(\lambda_{s}y+1-\lambda_{s})
\bigg\{ \mathbf{E}[\mathbf{1}[Q_{p}^{t}=0,Q_{s}^{t}=0]] \notag \\
& +\left[ \frac{f_{pd}+p_{a}f_{ps}(1-f_{pd})}{x} + (1-p_{a}f_{ps})(1-f_{pd}) \right]\notag \\
& \times \mathbf{E}[x^{Q_{p}^{t}}.\mathbf{1}[Q_{p}^{t}>0,Q_{s}^{t}=0]] \notag \\
& +\left[ \frac{p_{q}f_{sd}}{y} + 1 -p_{q}f_{sd} \right]
\mathbf{E}[y^{Q_{s}^{t}}.\mathbf{1}[Q_{p}^{t}=0,Q_{s}^{t}>0]]\notag \\
& +\left[ \frac{f_{pd}+p_{a}f_{ps}(1-f_{pd})}{x} + (1-p_{a}f_{ps})(1-f_{pd}) \right]\notag \\
& \times \mathbf{E}[x^{Q_{p}^{t}}y^{Q_{s}^{t}}.\mathbf{1}[Q_{p}^{t}>0,Q_{s}^{t}>0]] \bigg\}
\end{align}
where $\mathbf{1}[Z]$ is the indicator function of the discrete random variable $Z$, defined as
\begin{equation}
    \mathbf{1}[Z=z]=
    \begin{cases}
      1, & \text{w.p.}\ \mathbf{P}[Z=z] \\
      0, & \text{w.p.}\ \mathbf{P}[Z \neq z]
    \end{cases}
\end{equation}
Therefore, $\mathbf{E} \left[ \mathbf{1}[Z=z] \right]=\mathbf{P}[Z=z]$.
To explain the terms inside the braces of (\ref{expansion}), we analyze the $4$ possible combinations of the queue states, $Q_{p}^{t}$ and $Q_{s}^{t}$
\begin{itemize}
\item $Q_{p}^{t}=0, ~Q_{s}^{t}=0$

Since both queues are already empty, no departures occur, i.e., $Y_{p}^{t}=Y_{s}^{t}=0$. This explains the first term in the braces in (\ref{expansion}).

\item $Q_{p}^{t}>0, ~Q_{s}^{t}=0$

Clearly, no departures occur at $Q_{s}$ since it is empty, i.e., $Y_{s}^{t}=0$. At the PU side, it transmits a packet whenever it has a non-empty queue. Thus, $Y_{p}^{t}$ is given by
\begin{equation}
    Y_{p}^{t}=
    \begin{cases}
      1, & \text{w.p.}\ f_{pd}+p_{a}f_{ps}(1-f_{pd}) \\
      0, & \text{w.p.}\ (1-p_{a}f_{ps})(1-f_{pd})
    \end{cases}
\end{equation}
This states that a departure occurs at $Q_{p}$ if it is successfully received by the destination, or it is decoded by the SU and is admitted to its relay queue. Otherwise, no departures occur and the packet remains at $Q_{p}$ to be retransmitted in the next time slot. This gives the second term in the braces in (\ref{expansion}). 

\item $Q_{p}^{t}=0, ~Q_{s}^{t}>0$

The PU is idle, thus, $Y_{p}^{t}=0$. Then, the SU gains access to the system and transmits a packet. It randomly selects the source of this packet to be either $Q_{s}$ or $Q_{sp}$. Therefore, $Y_{s}^{t}$ is given by  
\begin{equation}
Y_{s}^{t}=
    \begin{cases}
      1, & \text{w.p.}\ p_{q}f_{sd} \\
      0, & \text{w.p.}\ 1-p_{q}f_{sd}
    \end{cases}
\end{equation}
This states that a departure occurs at $Q_{s}$ if it is selected to transmit, which happens w.p. $p_{q}$, and the transmitted packet is successfully decoded by the destination, which happens w.p. $f_{sd}$. Otherwise, no departures occur. This results in the third term in the braces in (\ref{expansion}).

\item $Q_{p}^{t}>0, ~Q_{s}^{t}>0$

Since the PU has the priority to transmit whenever it has packets, the SU is silent and $Y_{s}^{t}=0$. The PU transmits a packet and the queue $Q_{p}$ evolves exactly following the case of $Q_{p}^{t}>0, ~Q_{s}^{t}=0$ yielding the last term in the braces in (\ref{expansion}).
\end{itemize}
Taking the limit when $t \rightarrow \infty$ at both sides of (\ref{expansion}), we get
\begin{align} \label{G}
G(x,y)=(\lambda_{p}x+1-\lambda_{p})(\lambda_{s}y+1-\lambda_{s}) \notag \\
\times \frac{b(x,y)G(0,0)+c(x,y)G(0,y)}{yd(x,y)}
\end{align}
where
\begin{align} 
\begin{split}
& b(x,y)=xyp_{q}f_{sd}-xp_{q}f_{sd}
\end{split} \notag \\
\begin{split}
& c(x,y)= xp_{q}f_{sd}-y[f_{pd} +  p_{a}f_{ps}(1-f_{pd})] \\
& \phantom{c(x,y)=xp_{q}f_{sd}}
+ xy[f_{sd}+p_{a}f_{ps}(1-f_{pd})-p_{q}f_{sd}]
\end{split}\notag \\
\begin{split}
& d(x,y)=x-(\lambda_{p}x+1-\lambda_{p})(\lambda_{s}y+1-\lambda_{s})\times  \\
& \phantom{d(x,y)=} [f_{pd}+p_{a}f_{ps}(1-f_{pd})+x(1-p_{a}f_{ps})(1-f_{pd})]
\end{split}
\end{align}
From the definition of $G(x,y)$, note that
\begin{align}
\begin{split}
G(0,0)=\lim_{t \rightarrow \infty} \mathbf{E}[\mathbf{1}[Q_{p}^{t}=0,Q_{s}^{t}=0]]
\end{split}\notag \\
\begin{split}
G(x,0)=G(0,0)+\lim_{t \rightarrow \infty} \mathbf{E}[x^{Q_{p}^{t}}.\mathbf{1}[Q_{p}^{t}>0,Q_{s}^{t}=0]]
\end{split} \notag \\
\begin{split}
G(0,y)=G(0,0)+\lim_{t \rightarrow \infty} \mathbf{E}[y^{Q_{s}^{t}}.\mathbf{1}[Q_{p}^{t}=0,Q_{s}^{t}>0]]
\end{split} \notag \\
\begin{split}
& G(x,y)=G(x,0)+G(0,y)-G(0,0) \\
& \phantom{G(x,y)=} +\lim_{t \rightarrow \infty} 
\mathbf{E}[x^{Q_{p}^{t}}y^{Q_{s}^{t}}.\mathbf{1}[Q_{p}^{t}>0,Q_{s}^{t}>0]]
\end{split}
\end{align}
Along the lines of \cite{33}, $G(0,0)$ is evaluated using the normalization condition, $G(1,1)=1$, by taking the limit of (\ref{G}) when $(x,y)\rightarrow(1,1)$, which yields
\begin{equation}\label{G(0,0)}
G(0,0)\!\! = \!\! 
\frac{p_{q}f_{sd}[f_{pd} \!\! + \!\! p_{a}f_{ps}(1\!\! - \!\! f_{pd})\!\! - \!\! \lambda_{p}]
\!\! - \!\! \lambda_{s}[f_{pd}\!\! + \!\! p_{a}f_{ps}(1\!\! - \!\! f_{pd})]}
{p_{q}f_{sd}[f_{pd}+p_{a}f_{ps}(1-f_{pd})]}
\end{equation}
In the derivation of (\ref{G(0,0)}), we use the fact that
\begin{equation}
G(0,1)=\lim_{t \rightarrow \infty} \mathbf{P}[Q_{p}^{t}=0]=1-\frac{\lambda_{p}}{f_{pd}+p_{a}f_{ps}
(1-f_{pd})}
\end{equation}
To find $N_{s}$, we solve for $G_{y}(1,1)$. We evaluate the derivative of (\ref{G}) w.r.t. $y$, then take the limit of the result when $(x,y) \rightarrow (1,1)$. Applying L'Hopital's rule twice, we obtain an equation relating $G_{y}(1,1)$ to $G_{y}(0,1)$ as
\begin{equation}\label{E1}
G_{y}(1,1)=\lambda_{s}-1+\frac{p_{q}f_{sd}}{\lambda_{s}}G_{y}(0,1)
\end{equation} 
In order to characterize $G_{y}(0,1)$, we compute $\left.\frac{\partial G(y,y)}{\partial y} \right|_{y=1}$. We make use of the fact that $\left.\frac{\partial G(y,y)}{\partial y} \right|_{y=1}=N_{p}+N_{s}$, and $G_{y}(1,1)=N_{s}$. After some algebraic manipulation, we get
\begin{align}\label{E2}
G_{y}(1,1)&=
\frac{-(\lambda_{p}+\lambda_{s})^{2}+\lambda_{p}\lambda_{s}+\lambda_{p}+\lambda_{s}}
{f_{pd}+p_{a}f_{ps}(1-f_{pd})-\lambda_{p}-\lambda_{s}} -N_{p} \notag \\
&+\left[ \frac{f_{pd}+p_{a}f_{ps}(1-f_{pd})-p_{q}f_{sd}}{f_{pd}+p_{a}f_{ps}(1-f_{pd})-\lambda_{p}-\lambda_{s}} \right] G_{y}(0,1)
\end{align}
We can easily calculate $N_{p}$ by observing that $Q_{p}$ is a discrete-time $M/M/1$ queue with arrival rate $\lambda_{p}$ and service rate $\mu_{p}$. Thus, applying the Pollaczek-Khinchine formula \cite{PK}, $N_{p}$ is directly given by (\ref{pk}). Solving (\ref{E1}) and (\ref{E2}) together using the result obtained by (\ref{pk}), the term $G_{y}(0,1)$ is eliminated and $N_{s}$ is exactly given by (\ref{Ns}) in Theorem \ref{Thm2}. 

Next, we characterize the average delay experienced by the packets of the PU.
A PU's packet, if directly delivered to the destination, experiences the queueing delay at $Q_{p}$ only.
This happens w.p. $1-\epsilon=\frac{f_{pd}}{1-(1-p_{a}f_{ps})(1-f_{pd})}$, which is the probability that the packet is successfully decoded by the destination given that it is dropped from $Q_{p}$. Otherwise, if the transmission through the direct link between the PU and the destination fails, the packet is probably relayed through $Q_{sp}$ and, hence, experiences the total queueing delay at both $Q_{p}$ and $Q_{sp}$.
This happens w.p. $\epsilon$. Therefore, the average delay that a PU's packet experiences is given by
\begin{equation}\label{1}
D_{p}= (1-\epsilon)\tau_{p} + \epsilon (\tau_{p}+ \tau_{sp})= \tau_{p} + \epsilon \tau_{sp}
\end{equation}
where $\tau_{p}$ and $\tau_{sp}$ denote the average queueing delays at $Q_{p}$ and $Q_{sp}$, respectively. Since the arrival rates at $Q_{p}$ and $Q_{sp}$ are given by $\lambda_{p}$ and $\epsilon \lambda_{p}$, respectively. Then, applying Little's law yields
\begin{equation}\label{2}
\tau_{p}=N_{p}/ \lambda_{p}, \hspace{10mm} \tau_{sp}=N_{sp}/ \epsilon \lambda_{p}
\end{equation} 
Substituting (\ref{2}) in (\ref{1}) yields $D_{p}$ given by (\ref{Dp}). Provided that $N_{p}$ is already known by (\ref{pk}), the calculation of $D_{p}$ boils down to evaluating the average length of $Q_{sp}$, $N_{sp}$. As indicated earlier, the state of $Q_{sp}$ depends on that of $Q_{p}$, so we again employ the moment generating function approach to compute $N_{sp}$. Let 
$H(x,y)=\lim_{t \rightarrow \infty} \mathbf{E}[x^{Q_{p}^{t}}y^{Q_{sp}^{t}}]$ be defined as the moment generating function of the joint queue lengths of $Q_{p}$ and $Q_{sp}$. Using an analogous derivation employed to evaluate $G(x,y)$, we can write $H(x,y)$ as
\begin{equation}
H(x,y)=(\lambda_{p}x+1-\lambda_{p})\frac{b^{'}(x,y)G(0,0)+c^{'}(x,y)G(0,y)}{yd^{'}(x,y)}
\end{equation}  
where
\begin{align}
\begin{split}
& b^{'}(x,y)=x(1-p_{q})f_{sd}(y-1) 
\end{split}\notag \\
\begin{split}
& c^{'}(x,y)=x(1-p_{q})f_{sd}-yf_{pd}-y^{2}p_{a}f_{ps}(1-f_{pd}) \\
& \phantom{c^{'}(x,y)=} +xy[f_{pd}+p_{a}f_{ps}(1-f_{pd})-(1-p_{q})f_{sd}]
\end{split}\notag \\
\begin{split}
& d^{'}(x,y)=x-(\lambda_{p}x+1-\lambda_{p}) \times \\
& \phantom{d^{'}(x,y)=}
[f_{pd}+yp_{a}f_{ps}(1-f_{pd})+x(1-p_{a}f_{ps})(1-f_{pd})]
\end{split}
\end{align}
Following the same footsteps of the approach employed to evaluate $N_{s}$, $N_{sp}$ is shown to be given by (\ref{Nsp}). 
\end{proof}

\begin{figure}[t]
\begin{center}
\includegraphics[width=1\columnwidth , height=0.68\columnwidth]{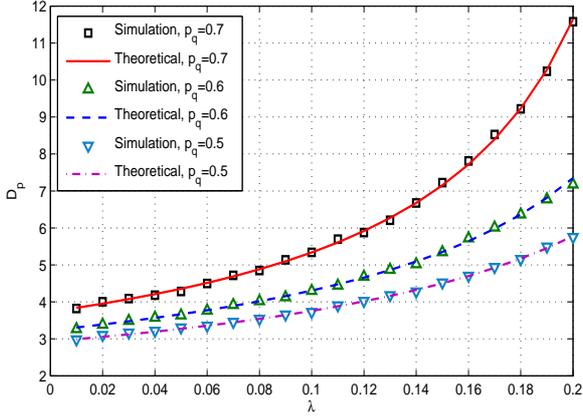}
\caption{Average delay for the PU's packets at $p_{a}=1$ and different values of $p_{q}$}\label{Fig6}
\end{center}
\vspace{-7mm}
\end{figure}

So far, we have obtained closed-form expressions for the delay at both users, i.e., $D_{p}$ and $D_{s}$. Next, we proceed with analyzing the behavior of $D_{p}$ and $D_{s}$ in response to variations in each of $p_{q}$ and $p_{a}$ individually.

\begin{lemma}\label{Lemma3}
Under the proposed randomized service policy with probabilistic relaying, if the system is stable at a fixed operating point $(\lambda_{p},\lambda_{s})$, the average delay experienced by the packets of the PU, $D_{p}$, is a monotonically increasing function in $p_{q}$, while the average delay encountered by the packets of the SU, $D_{s}$, decreases monotonically with $p_{q}$.
\end{lemma}
\begin{proof}
The logic behind proving this result goes as follows. Instead of going forward with deriving the delay at the PU and SU, $D_{p}$ and $D_{s}$, respectively, w.r.t. $p_{q}$, we directly make use of the result presented and proven in Lemma \ref{Lemma1}. We rely on the fact that the delay at both the PU and SU at a fixed $(\lambda_{p},\lambda_{s})$ depends on the difference between the operating values of arrival rates, $\lambda_{p}$ and $\lambda_{s}$, and their corresponding maximums given by (\ref{Q21_modified}) and (\ref{Q22}), respectively. Analyzing the behavior of the delay at the PU's side first, we realize that increasing $p_{q}$ decreases the maximum achievable $\lambda_{p}$ resulting in shrinking the distance between the operating $\lambda_{p}$ and the stability region's boundary for the PU, i.e., the maximum achievable arrival rate at the PU. Therefore, the delay of the PU's packets, $D_{p}$, increases with the increase of $p_{q}$ until it reaches infinity when the maximum achievable arrival rate at the PU, given by (\ref{Q21_modified}), coincides with the operating $\lambda_{p}$. This is attributed to the critical stability of the system in that case.

Using the same rationale at the SU side, it has been established by Lemma \ref{Lemma1} that increasing $p_{q}$ increases the maximum achievable $\lambda_{s}$ for every given $\lambda_{p}$. Thus, as $p_{q}$ increases, the distance between the operating value of $\lambda_{s}$ and its maximum achievable value increases. Therefore, the operating point is pushed deeper in the stability region from the SU's point of view and accordingly, $D_{s}$ decreases.     
\end{proof} 

\begin{figure}[t]
\begin{center}
\includegraphics[width=1\columnwidth , height=0.68\columnwidth]{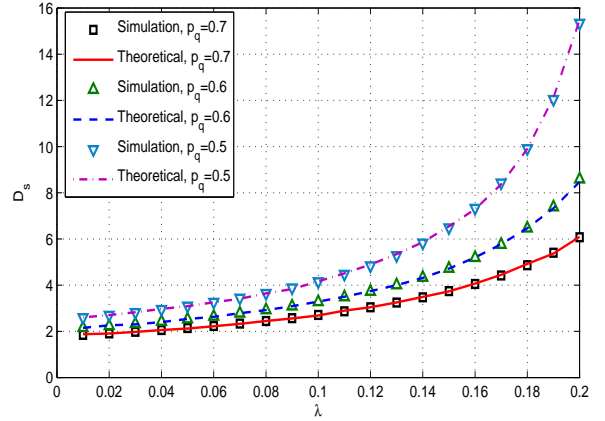}
\caption{Average delay for the SU's packets at $p_{a}=1$ and different values of $p_{q}$}\label{Fig7}
\end{center}
\vspace{-7mm}
\end{figure}

\begin{lemma}\label{Lemma4}
Under the proposed randomized cooperative policy with probabilistic relaying, if the system is stable at a fixed operating point $(\lambda_{p},\lambda_{s})$, the average delay experienced by the packets of the PU, $D_{p}$, is a monotonically decreasing function in $p_{a}$ if $p_{q}$ lies in the interval $\left( 0,1-\frac{f_{pd}}{f_{sd}}\right)$, and increases monotonically with $p_{a}$ if $p_{q}$ lies in the interval $\left(1-\frac{f_{pd}}{f_{sd}},1\right)$. However, the average delay encountered by the packets of the SU, $D_{s}$, decreases monotonically with $p_{a}$, irrespective of the choice of $p_{q}$. 
\end{lemma} 
\begin{proof}
We follow the same rationale adopted to prove Lemma \ref{Lemma3}. The cornerstone that we rely on is the fact that the delay at both the PU and SU at a fixed $(\lambda_{p},\lambda_{s})$ depends on the difference between the operating values of arrival rates, $\lambda_{p}$ and $\lambda_{s}$, and their corresponding maximums given by (\ref{Q21_modified}) and (\ref{Q22}), respectively. Delay decreases as this difference increases and vice versa. Using this fact along with our knowledge in Lemma \ref{Lemma2}, the proof of this result directly follows.
\end{proof} 
 
\section{Numerical Results}\label{numerical results}
In this section, we analyze the performance of the system under the proposed policy. Extensive simulations are conducted to validate the closed-form expressions obtained for the average delay experienced by the packets of the PU and SU. Furthermore, we characterize and analyze fundamental tradeoffs that arise at both users such as the delay-throughput tradeoff, as well as the tradeoff between the PU and SU in terms of delay. It is shown that the system's performance is flexibly tuned using the parameters $p_{q}$ and $p_{a}$. Moreover, in an attempt to show the potential of employing the proposed policy, we compare the performance of the system under the proposed policy with that of existing schemes. 

\begin{figure}[t]
\begin{center}
\includegraphics[width=1\columnwidth , height=0.68\columnwidth]{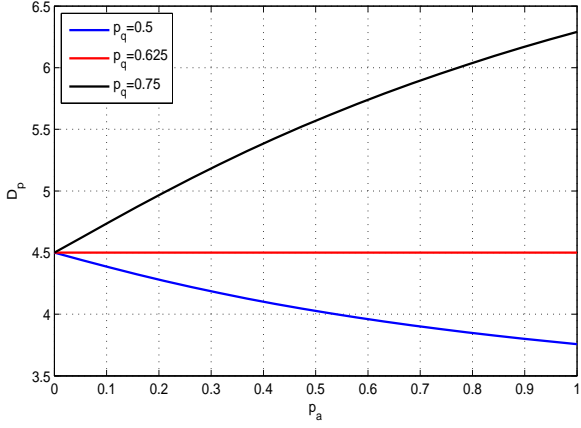}
\caption{Average delay for the PU's packets versus $p_{a}$ for different values of $p_{q}$}\label{Fig8}
\end{center}
\vspace{-7mm}
\end{figure}

In Figs.~\ref{Fig6} and \ref{Fig7}, we plot the average delay experienced by the packets of the PU and SU, respectively, versus $\lambda$, where we choose $\lambda_{p}=\lambda_{s}=\lambda$ for ease of exposition. We fix $p_{a}=1$ and vary $p_{q}$. These figures aim at validating the obtained closed-form expressions for the delays by comparing them to the simulation results. In the simulation, the delay is averaged over $10^{6}$ time slots. It can be viewed that the results obtained through simulations coincide with the results of the closed-form expressions derived in Theorem \ref{Thm2}. This validates the soundness of the mathematical model and the moment generating function approach. Moreover, at a given $\lambda$, when $p_{q}$ increases, $D_{p}$ is shown to increase, while $D_{s}$ decreases. This matches the result stated by Lemma \ref{Lemma3}.

The result obtained in Lemma \ref{Lemma4} is clarified via Figs. \ref{Fig8} and \ref{Fig9}, where we plot the average delay of the PU's packets, given by (\ref{Dp}), and the average delay of the SU's packets, given by (\ref{Ds}), versus $p_{a}$, respectively, for different values of $p_{q}$ at an operating point $\lambda_{p}=\lambda_{s}=0.1$. It is shown in Fig. \ref{Fig8} that $D_{p}$ decreases monotonically with $p_{a}$ as long as $p_{q}<1-\frac{f_{pd}}{f_{sd}}$. Conversely, it increases monotonically with $p_{a}$ for $p_{q}>1-\frac{f_{pd}}{f_{sd}}$ while remaining constant at $p_{q}=1-\frac{f_{pd}}{f_{sd}}$. The intuitive explanation behind this behaviour is the same as that given in our comments on Fig. \ref{Fig4}. If $(1-p_{q})f_{sd} > f_{pd}$, then the PU's average delay is reduced via cooperation, i.e., having more packets getting relayed by the SU reduces the delay at the PU, since the success probability over the relay-destination link is greater than that of the PU-destination link. However, if $(1-p_{q})f_{sd} < f_{pd}$, then it is always in the interest of the PU to retransmit its lost packets rather than getting them relayed via the SU, i.e., rejecting more packets at $Q_{sp}$ reduces the delay at the PU. 
In Fig. \ref{Fig8}, it can be noticed that at a fixed $p_{a}$, the PU's packets experience lower delay at lower values of $p_{q}$, which is the result obtained in Lemma \ref{Lemma3} and referred to in the comments on Fig. \ref{Fig6}. In addition, one can notice that the degradation in the PU's delay with the increase of $p_{q}$ decreases at lower values of $p_{a}$. These results again emphasize the importance of the parameter $p_{a}$ in tuning the system's performance.

\begin{figure}
\begin{center}
\includegraphics[width=1\columnwidth , height=0.68\columnwidth]{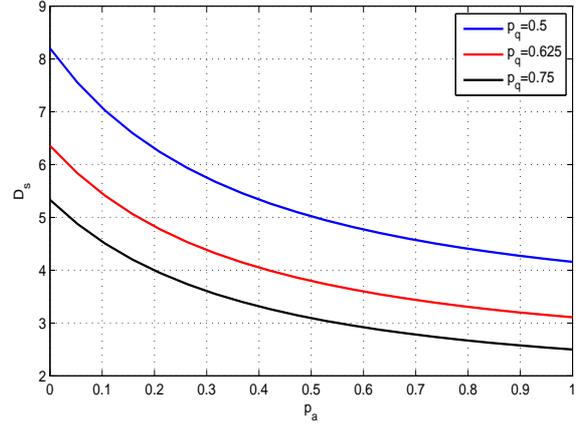}
\caption{Average delay for the SU's packets versus $p_{a}$ for different values of $p_{q}$}\label{Fig9}
\end{center}
\vspace{-7mm}
\end{figure}

On the other hand, we plot in Fig. \ref{Fig9} the average delay of the SU's packets, $D_{s}$, versus $p_{a}$. It can be depicted that $D_{s}$ decreases monotonically with $p_{a}$ irrespective of the choice of $p_{q}$. This assures our previous conjecture that the SU is always benefiting from increasing $p_{a}$, which is a direct consequence of the increase in the availability of time slots in which the PU's queue is empty, since at higher values of $p_{a}$, more packets are dropped from $Q_{p}$. Thus, the SU's packets are more likely to be transmitted. Moreover, at a fixed $p_{a}$, Fig. \ref{Fig9} shows that the SU's packets experience lower delay at higher values of $p_{q}$, which again emphasizes the result obtained in Lemma \ref{Lemma3}.

We realize from Figs. \ref{Fig8} and \ref{Fig9} that a delay tradeoff between the PU and SU arises due to introducing the admission control parameter, i.e., $p_{a}$, in the range of $p_{q}$ values that belong to the interval $\left[ 1-\frac{f_{pd}}{f_{sd}},1 \right)$. Since in this interval of $p_{q}$ values, cooperation becomes in favour of the SU's delay at the expense of the PU's delay, there exists a conflict of interest between both users. Fig. \ref{Fig10} depicts this tradeoff, where we plot $D_{p}$ versus $D_{s}$ at an operating point of $\lambda_{p}=\lambda_{s}=0.1$ at different $p_{q}$ values. In that plot, every point $(D_{s},D_{p})$ corresponds to a certain value of $p_{a}$. Thus, we can now see a dimension of the benefit of tuning the admission control parameter, $p_{a}$. In Fig. \ref{Fig10}, as $p_{a}$ increases, $D_{s}$ decreases, while $D_{p}$ increases if $p_{q}>1-\frac{f_{pd}}{f_{sd}}$ or remains constant at $p_{q}=1-\frac{f_{pd}}{f_{sd}}$.   

Next, we characterize a fundamental tradeoff that arises between the average delay and the throughput at both the PU and SU. Intuitively, when a node needs to maintain a higher throughput, it loses in terms of the average delay encountered by its packets. Given that the system is stable, the node's throughput equals its packet arrival rate. Thus, increased throughput means injecting more packets into the system which yields a higher delay. In Fig. \ref{Fig11}, we illustrate the delay-throughput tradeoff at the PU. Note that, given the stability of the system, the throughput of the PU equals $\lambda_{p}$. We fix the value of $\lambda_{s}$ at $0.2$. Then, at every $\lambda_{p}$, we formulate and solve the following optimization problem

\begin{align}\label{optimization}
& \underset{(p_{q},p_{a})}{\text{minimize}}
& & D_{p} \notag \\
& \text{subject to}
& & \lambda_{p}< \frac{f_{sd}(1-p_{q})[f_{pd} + p_{a}f_{ps}(1-f_{pd})]}
{f_{sd}(1-p_{q})+p_{a}f_{ps}(1-f_{pd})}, \notag \\
&&& \lambda_{s}<p_{q}f_{sd} \left[ 1 - \frac{\lambda_{p}}{f_{pd}+p_{a}f_{ps}(1-f_{pd})} \right], \notag \\
&&& 0<p_{q}<1, \notag \\
&&& 0\leq p_{a} \leq 1.
\end{align} 
Thus, we solve for the optimal value of $(p_{q},p_{a})$ that minimizes $D_{p}$ while simultaneously keeping the system stable at $(\lambda_{p},\lambda_{s})$. Lemmas \ref{Lemma3} and \ref{Lemma4} are of fundamental importance in the approach employed to obtain the solution. Initially, we characterize the feasible set, which is the set of $(p_{q},p_{a})$ values that satisfy the constraints, i.e., keeping the system stable at $(\lambda_{p},\lambda_{s})$. After simple algebraic manipulation on the constraints, we obtain lower and upper bounds on $p_{q}$, $p_{q}^{(l)}$ and $p_{q}^{(u)}$, respectively, as functions of $p_{a}$ as follows  
\begin{align} 
p_{q}^{(l)}=\frac{\lambda_{s} [f_{pd}+p_{a}f_{ps}(1-f_{pd})]}{f_{sd}[f_{pd}+p_{a}f_{ps}(1-f_{pd})-\lambda_{p}]} \label{lb} \\ 
p_{q}^{(u)}=1-\frac{\lambda_{p}p_{a}f_{ps}(1-f_{pd})}{f_{sd}[f_{pd}+p_{a}f_{ps}(1-f_{pd})-\lambda_{p}]}
\end{align}
\begin{figure}
\begin{center}
\includegraphics[width=1\columnwidth , height=0.68\columnwidth]{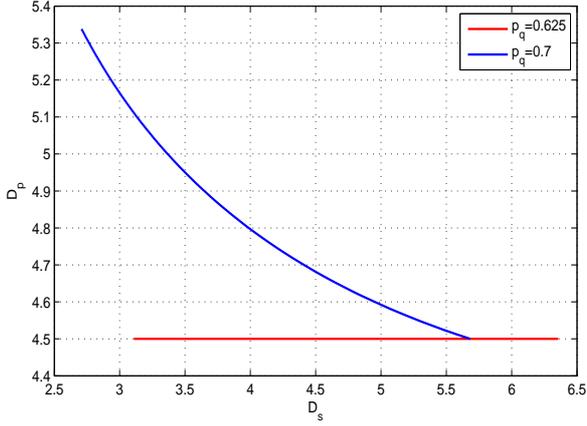}
\caption{PU-SU delay tradeoff}\label{Fig10}
\end{center}
\vspace{-7mm}
\end{figure}
Lemma \ref{Lemma3} establishes that $D_{p}$ is monotonically increasing in $p_{q}$, irrespective of the choice of $p_{a}$. Therefore, to minimize $D_{p}$, we go for the minimum over all feasible values of $p_{q}$ which is obtained through minimizing (\ref{lb}). We realize that $p_{q}^{(l)}$ decreases monotonically with $p_{a}$. This is easily shown through taking the derivative of (\ref{lb}) w.r.t. $p_{a}$ that yields
\begin{align}
\frac{\partial p_{q}^{(l)}}{\partial p_{a}}= \frac{-\lambda_{s}\lambda_{p}f_{sd}f_{ps}(1-f_{pd})}{[f_{sd}(\mu_{p}-\lambda_{p})]^{2}}
\end{align}
which is clearly negative definite, whereby $\mu_{p}$ is given by (\ref{mu_p}). Thus, the minimum over all feasible values of $p_{q}$, let it be denoted by $p_{1_{\text{min}}}^{(l)}$, is obtained via evaluating (\ref{lb}) at $p_{a}=1$. Proceeding with the solution, we make use of the result obtained in Lemma \ref{Lemma4}, and shown by Fig. \ref{Fig8}, that defined the region of $p_{q}$ values at which the PU can benefit from cooperation, i.e., $D_{p}$ decreases with increasing $p_{a}$. Comparing $p_{1_{\text{min}}}^{(l)}$ to the threshold value of $1-\frac{f_{pd}}{f_{sd}}$, the optimal solution is    decided which is either to cooperate or not to cooperate. If $p_{1_{\text{min}}}^{(l)} \leq 1-\frac{f_{pd}}{f_{sd}}$, the optimal value of $(p_{q},p_{a})$ is at $(p_{1_{\text{min}}}^{(l)},1)$, otherwise, no cooperation yields a lower $D_{p}$. This solution relies on the fact that if there is no single value of $p_{q}$ that makes the system stable at $p_{a}=1$, the problem is infeasible. 

\begin{figure}
\begin{center}
\includegraphics[width=1\columnwidth , height=0.68\columnwidth]{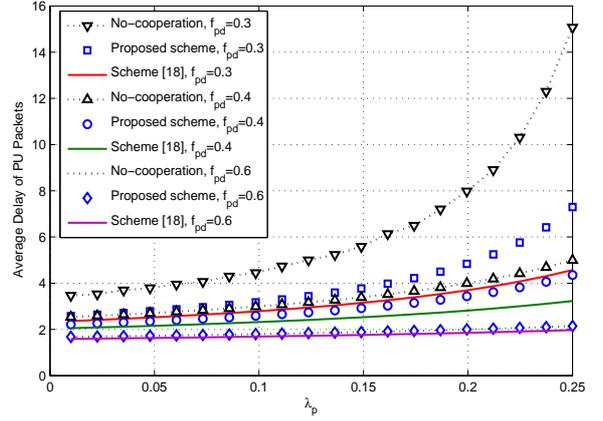}
\caption{Delay-throughput tradeoff at the PU}\label{Fig11}
\end{center}
\vspace{-7mm}
\end{figure}

In Fig. \ref{Fig11}, we plot the optimal PU delay for the proposed scheme, the no-cooperation baseline and the conventional relaying that gives strict priority to the relay queue. System parameters are the same as indicated previously in Section \ref{stable throughput region} except for $f_{pd}$, whereby we show our results at three values of $f_{pd}$, specifically, $0.3$, $0.4$ and $0.6$. It can be realized that cooperation benefits decreases with the increase of $f_{pd}$. We also notice that the performance of the proposed cooperative scheme in terms of the PU's delay coincides with the no-cooperation case at $f_{pd}=0.6$, which means that the PU's interest in cooperation vanishes. The intuitive explanation behind this result goes as follows. As long as the quality of the direct link between the PU and the destination increases, the benefits of cooperation decreases, since the probability of packet success in the direct link increases. This also explains why the threshold of $p_{q}$ values below which cooperation is beneficial to the PU is given by $1-\frac{f_{pd}}{f_{sd}}$. The intuition behind this behavior has previously been stressed on in Section \ref{stable throughput region} and revisited in Section \ref{numerical results}. It is worth noting that the performance of \cite{Ephremedis} in terms of the PU's delay is superior to that of the proposed scheme. This expected result is attributed to the strict priority given in \cite{Ephremedis} to the relay queue. However, in our scheme, we randomly select either the relay queue or the queue of own packets at the SU. Thus, we open room for trading the PU delay for enhanced SU delay and vice versa.

We now turn to the SU side investigating the same delay-throughput tradeoff. We fix the value of $\lambda_{p}$ at $0.2$. Then, at every $\lambda_{s}$, we revisit the problem formulated in (\ref{optimization}) with the objective of minimizing the SU's delay instead of the PU's delay, i.e., minimizing $D_{s}$ instead of $D_{p}$. Solving the problem numerically, we conjecture that it boils down to the solution presented in \cite{Ashour}, where $p_{a}=1$ is always optimal. The result obtained in Lemma \ref{Lemma4} and shown in Fig. \ref{Fig9} suggests that cooperation is always in favour of the SU, i.e., increasing $p_{a}$ reduces $D_{s}$. This explains why the optimal value is always obtained at $p_{a}=1$. The resulting delay-throughput curves for the proposed policy as well as for \cite{Ephremedis} are shown in Fig. \ref{Fig12}. We avoided plotting the no-cooperation baseline case to have a clear view for the comparison between the plotted policies, since the no-cooperation performance is way worse than both. It can be viewed that at the SU, the best achievable performance of the system under the proposed policy in terms of the average delay at the SU side, is superior to the performance of the system under the policy proposed in \cite{Ephremedis}.      

\begin{figure}
\begin{center}
\includegraphics[width=1\columnwidth , height=0.68\columnwidth]{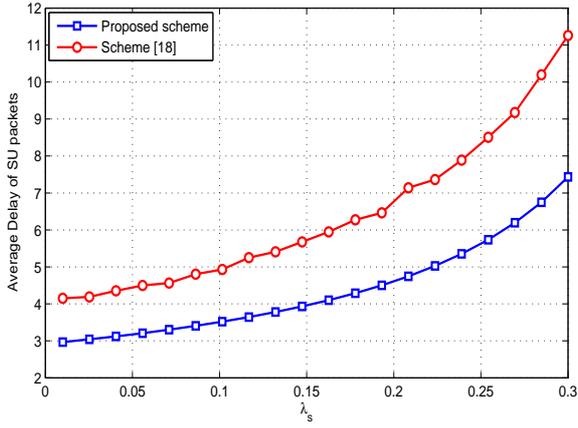}
\caption{Delay-throughput tradeoff at the SU}\label{Fig12}
\end{center}
\vspace{-7mm}
\end{figure}
\section{Conclusion}\label{conclusion}
In this work, we propose a cooperative policy with randomized service, whereby the SU probabilistically selects to serve either the queue of its own data or the relay queue w.p. $p_{q}$ and $(1-p_{q})$, respectively, upon detecting a spectral hole. Moreover, we introduce an admission control parameter, $p_{a}$, that acts as a flow controller to the traffic coming to the relay queue from the PU. A comprehensive analysis of the system's performance metrics such as stable throughput and delay is introduced, where we thoroughly investigate the effect of tuning both $p_{q}$ and $p_{a}$ on the performance of the system. Our major findings are represented in the following. The complete stable throughput region of the system obtained via taking the union over all possible values of $(p_{q},p_{a})$ is shown to strictly contain the stability region of the no-cooperation baseline. Thus, cooperation is shown to expand the stability region of the system. In addition, it has been shown that increasing $p_{q}$ is always in favour of the SU as opposed to the PU in terms of both throughput and delay. This behaviour is irrelevant to the choice of $p_{a}$. However, introducing $p_{a}$ at the relay queue enables us to clearly define the region of $p_{q}$ values at which cooperation is beneficial to the PU. It has been established that as long as $p_{q}<1-\frac{f_{pd}}{f_{sd}}$, cooperation enhances the PU's throughput and reduces its delay. On the contrary, if $p_{q}>1-\frac{f_{pd}}{f_{sd}}$, no cooperation becomes better from the PU's point of interest. This suggests using $p_{a}$ as a switch with which we decide to cooperate or not to cooperate to optimize the performance of the PU depending on the value of $p_{q}$ as well as the channel qualities, i.e., $f_{pd}$ and $f_{sd}$. Finally, we characterize and analyze the delay-throughput tradeoff at the PU and SU, as well as the tradeoff that arises between the delays of the PU and SU. The latter tradeoff shows the potential of using $p_{a}$ as a flow controller at the relay queue.

\bibliographystyle
{IEEEtran}
\bibliography{IEEEabrv,Probabilistic_relaying}
\end{document}